\documentclass[%
 reprint,
]{revtex4-2}
\usepackage{amsmath,amssymb,amsthm,mathtools}
\usepackage{tikz-cd}
\usepackage{bm}
\usepackage{graphicx} % To include figures
\usepackage{tikzit}
\usepackage[normalem]{ulem}
\usepackage{dsfont}

\usepackage{tikz}
\usepackage{placeins}
\usepackage{soul,xcolor}
\usepackage[utf8]{inputenc}
\usepackage{lipsum,textcomp}
\usepackage{blkarray}
\usepackage{comment}

\makeatletter
\theoremstyle{plain} 
\newtheorem{theorem}{Theorem}%[section] 

\theoremstyle{definition}               % Definitions will share numbering with theorems

\theoremstyle{plain}
\newtheorem{Proposition}[theorem]{Proposition}

\DeclareMathOperator{\Fun}{Fun}
\def\CA {{\cal A}}
\def\CM {{\cal M}}

\def\CN {{\cal N}}

\def\CP {{\cal P}}
\def\CQ {{\cal Q}}
\def\CC {{\mathcal C}}
\def\CD {{\mathcal D}}

\def\CX {{\mathcal X}}
\def\CY {{\mathcal Y}}
\def\CO {{\mathcal O}}
\def\CZ {{\mathcal Z}}

\newcommand{\calC}{{\mathcal C}}
\newcommand{\calD}{{\mathcal D}}

\newcommand{\calW}{{\mathcal W}}
\newcommand{\calZ}{{\mathcal Z}}

\newcommand{\Tr}{{\textrm{Tr}}}

\makeatother

\definecolor{Mathematica1}{rgb}{0.368417, 0.506779, 0.709798}
\definecolor{Mathematica2}{rgb}{0.880722, 0.611041, 0.142051}
\definecolor{Mathematica3}{rgb}{0.560181, 0.691569, 0.194885}

\definecolor{MMcolor}{RGB}{0,153,153}

\usepackage[pagebackref % this puts links to the page numbers where refs appear
]{hyperref}

\hypersetup{colorlinks=true, 
linkcolor=green!50!black,
citecolor=red!40!orange!90!black,
filecolor=OliveGreen, 
urlcolor=Mathematica2!80!orange!95!black,
filebordercolor={.8 .8 1}, 
urlbordercolor={.8 .8 0}}

\usepackage{multirow}

\usepackage[shortlabels]{enumitem}
\usepackage{wrapfig}
\usepackage{setspace}
\usepackage[top=25truemm,bottom=30truemm,left=22truemm,right=22truemm]{geometry}

\definecolor{BSorange}{RGB}{140,50,0}

\begin{document}
\title{\bfseries\Large An Algebraic Theory of Gapped Domain Wall Partons}
\author{Matthew Buican,$^1$ Roman Geiko,$^2$ Milo Moses,$^3$ and Bowen Shi$^{4,\,5,\,6}$}
\affiliation{$^1$CTP and Department of Physics and Astronomy, Queen Mary University of London, London E1 4NS, UK\\ $^2$Department of Physics and Astronomy, University of California, Los Angeles, 90095, CA, USA \\ $^3$Department of Mathematics, California Institute of Technology, Pasadena, CA 91125, USA\\ $^4$Physics Department, University of Illinois at Urbana-Champaign, Urbana, Illinois 61801, USA\\
$^5$Department of Computer Science, University of California, Davis, CA 95616, USA\\
$^6$Department of Physics, University of California at San Diego, La Jolla, CA 92093, USA}
\date{\today}

\begin{abstract}
The entanglement bootstrap program has generated new quantum numbers associated with degrees of freedom living on gapped domain walls between topological phases in two dimensions. Most fundamental among these are the so-called \lq\lq parton" quantum numbers, which give rise to a zoo of composite sectors. In this note, we propose a categorical description of partons. Along the way, we make contact with ideas from generalized symmetries and SymTFT.
\end{abstract}

\maketitle

\section{Introduction}

Starting with the realization that anyons in two-dimensional topologically ordered systems are described by unitary modular tensor categories (uMTCs) \cite{moore1991nonabelions,kitaev2006anyons}, there has been a productive dialogue between category theory and condensed matter physics whereby physical structures in the world of topological order have been interpreted and characterized in terms of abstract algebraic structures. A developing research program known as the {\em entanglement bootstrap program} seeks to recover these categorical classifications of topological order from first principles \cite{ShiKim2021, shi2020fusion, kim2024classifying,immersion2023}.

Recently, one of the present authors in collaboration with Kim extended the entanglement bootstrap program to the case of gapped domain walls between topological phases \cite{ShiKim2021}. Surprisingly, this first-principles approach uncovered a phenomenon which seems to have been absent in the literature up to this point. The main idea is that topological domain wall defects have two internal quantum numbers ---classified by so-called \lq\lq $N$-type" and \lq\lq$U$-type" partons---that partially determine their overall defect type. Moreover, the entanglement bootstrap partons determine the quantum numbers of various composite sectors analogously to the way in which particle physics partons organized the \lq\lq particle zoo" discovered in the mid 20th century. The precise definition, meaning, and practical implications of the parton quantum numbers will be discussed throughout this note.

This state of affairs leaves the theory in an interesting situation. Parton sectors are now understood at the level of first principles, but they have not been interpreted with the machinery of category theory. It is the purpose of this note to remedy the situation by proposing a categorical interpretation of parton sectors. We then make contact with certain quantities discussed in the recent generalized symmetries and SymTFT literature. Finally, in a set of appendices, we discuss examples and some generalizations of the formalism we develop in the main text. We also give a more elaborate review of the entanglement bootstrap, extend some of its ideas, and reinterpret certain additional entanglement bootstrap computations using our formalism.

\begin{figure}
\includegraphics[width=8cm]{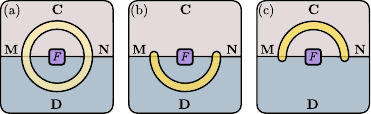}
\caption{The relevant geometries for defining parton sectors. (a) The domain wall defect type, ${\bf F}$, can be measured on an annulus surrounding it. (b-c) The $U$-type ($N$-type) parton sector corresponding to ${\bf F}$ is measured on a half-annulus below (above) the defect.}
\label{parton-diagrams}
\end{figure}

\section{Parton sectors}

In this section we review the definition of parton sectors \cite{ShiKim2021} and the minimal relevant background from the entanglement bootstrap program. As a first step, we recall the entanglement bootstrap approach to bulk anyons \cite{shi2020fusion, shi2019characterizing}.

Suppose we are given a state, $\rho$, corresponding to a topological phase, $\bf C$, with a single anyon in the bulk. The anyon type can be determined by looking at the reduced density matrix on an annulus surrounding the anyon. That is, given two states $\rho$ and $\rho'$ hosting single anyons in $\bf C$ at a distinguished point, $\rho$ and $\rho'$ have the same reduced density matrix on the annulus if and only if the anyons at the distinguished point are of the same type. In this way, anyon types are in bijection with a class of well-behaved states on the annulus \footnote{In the language of the entanglement bootstrap program these well-behaved states are called {\em information convex states}.}. At a basic level, this discussion re-packages the fact that the anyon type in a region is the maximal topological information accessible to measurements localized on an annulus around the region.

It is natural to extend this characterization to domain wall defects. Here, one considers a state, $\rho$, on a space divided into half-planes such that the top half-plane is in phase $\bf C$ and the bottom half-plane is in phase $\bf D$. On the interface between these half-planes, there is a gapped domain wall that is divided into two half-lines, one of type $\bf M$ and one of type $\bf N$. At the intersection of the half-lines, there is a topological defect of type ${\bf F}$ (see Fig. \ref{parton-diagrams}~(a)). As in the case of a bulk anyon, the defect type is characterized by the reduced density matrix on an annulus surrounding it.

The insight behind parton sectors is to observe that, by considering closely related regions, one can resolve the above domain wall defect superselection sectors into finer superselection sectors (see Fig. \ref{parton-diagrams}~(b-c)). Namely, one can consider half-annular regions around the defect that are disconnected on one side of the domain wall. The reduced density matrix on these regions does not fully determine ${\bf F}$, but it does give non-trivial topological information. The information accessible from the half-annulus on the upper (lower) side of the domain wall is known as the $N$-type ($U$-type) parton sector corresponding to the domain wall defect.

\section{The pinching trick}\label{sec:pinching}

In this section, we introduce the {\em pinching trick}, the geometric insight behind our category-theoretic interpretation of parton sectors \footnote{The pinching trick is named so as to distinguish it from the {\em folding trick}, a separate geometric technique for analyzing gapped domain walls.}. Suppose we are given a state corresponding to a single domain wall defect as in Fig. \ref{parton-diagrams}. The pinching trick is a series of topological manipulations of this state and a distinguished region (see Fig. \ref{pinching-trick}). Our initial configuration consists of a straight domain wall and a distinguished region corresponding to a half-annulus on one side of the defect. We then fold the ray corresponding to $\bf M$ by $90$ degrees at two different points to produce a finite-width slice of phase $\bf C$ bounded by $\bf N$ and $\bf M^*$. The distinguished region is then moved until it becomes a disk connecting the two parallel sides of the slice. The segment of the slice containing the defect is finally moved far away from the distinguished region.

\begin{figure}
\includegraphics[width=7.5cm]{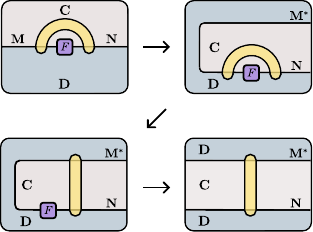}
\caption{The pinching trick. In a series of manipulations, the geometry defining parton sectors is mapped to the geometry of a disk overlapping a composite domain wall.}  
\label{pinching-trick}
\end{figure}

The resulting configuration is an interface between two gapped domain walls. Typically, gapped domain walls have non-abelian fusion rules. Clearly, the information contained in the region between the two domain walls is exactly the information describing which fusion channel the two domain walls take when they are brought together. There is a well-developed algebraic machinery for analyzing the fusion of domain walls that gives a handle on the information contained in parton sectors \cite{Kitaev_2012, huston2023composing}.

We now further explain the pinching trick and the physical relevance of parton sectors through a gedanken experiment (see Fig. \ref{gendanken}). We imagine an experimental setup where a topologically ordered state is prepared with the following geometry. The state is inhomogeneous, consisting of two phases separated by a domain wall. On the domain wall, there are two well-separated domain wall defects. After some distance in either direction, the domain wall turns and the two sides of the domain wall approach each other in a junction region. Due to interaction effects between these nearby sections of the domain wall, we expect that the state will decohere~\cite{Zurek2003} and be projected onto a state corresponding to a well-defined fusion channel between the two halves of the domain wall.

\begin{figure}
\includegraphics[width=4.5cm]{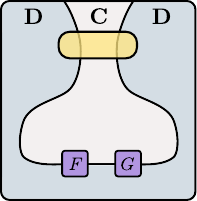}
\caption{An inhomogeneous topological state with two bulk phases, ${\bf C}$ and ${\bf D}$, separated by a curved domain wall. Near the junction region where the sides of the domain wall approach each other, local measurements can probe partial information on the overall fusion channel of the domain wall defects ${\bf F}$ and ${\bf G}$.}
\label{gendanken}
\end{figure}

Without the junction region, the system hosts delocalized logical degrees of freedom in the unconstrained fusion space between the two domain wall defects. On the other hand, in the presence of the junction, these delocalized degrees of freedom are forced to take values in a sector in which the fusion of the two defects has a definite $N$-type parton quantum number. Therefore, the $N$-type parton acts as a superselection sector for the domain wall defects. 

In principle, the parton sector measurements made possible by controlling these sorts of junctions should be topologically protected. As such, they can be included in the topological quantum computing toolkit. This idea is especially relevant in light of the fact that some models of quantum computation with gapped domain walls are only known to be universal when additional measurement primitives are included \cite{cong2017universal} and that gapped defects can be coherently manipulated on quantum simulation platforms \cite{iqbal2024qutrit, xu2023digital}.

\section{Partons via categories}\label{Sec:PartonsViaCats}

In this section, we propose an algebraic characterization of partons based on Kitaev and Kong's theory of gapped domain walls for two-dimensional topologically ordered systems~\cite{Kitaev_2012}. As such, we restrict ourselves to states in the same topological phase as Levin-Wen string nets. The outline of Kitaev and Kong's theory is summarized in Table \ref{Kitaev-Kong}.

From this picture, we can interpret our physical setup in terms of categories as follows. The bulk phases, $\bf C$ and $\bf D$, correspond to unitary fusion categories $\mathcal{C}$ and $\mathcal{D}$. The domain walls, $\bf N$ and $\bf M$, correspond to indecomposable $(\mathcal{C},\mathcal{D})$-bimodule categories $\mathcal{M}$ and $\mathcal{N}$. The defect, $\bf F$, corresponds to an irreducible $(\mathcal{C},\mathcal{D})$-bimodule functor $F:\mathcal{M}\to\mathcal{N}$.

\begin{table}[h]
\begin{tabular}{ |p{3.7cm}||p{3.7cm}|  }
 \hline
 Levin-Wen model& Category theory\\
 \hline
 Bulk Levin-Wen model   & Unitary fusion category\\
 Gapped domain wall & Bimodule category  \\
 Domain wall defect & Bimodule functor \\
 Composition of domain walls  & Relative tensor product \\
 \hline
\end{tabular}
\caption{The Kitaev-Kong dictionary for domain walls.}
\label{Kitaev-Kong}
\end{table}

We can now give an algebraic perspective on the pinching trick. To that end, note that once we rotate the left half of the domain wall to produce the parallel configuration in Fig. \ref{pinching-trick}, the parallel walls can be treated as a composite domain wall, $\mathcal{M}^* \boxtimes_{\mathcal{C}} \mathcal{N}$. Here, $\mathcal{M}^*$ denotes the opposite category, $\mathcal{M}^{\rm op}$, treated as a $(\mathcal{C},\mathcal{D})$-bimodule, which physically corresponds to reversing the orientation of $\bf M$, and $\boxtimes_{\mathcal{C}}$ is the relative Deligne tensor product over $\mathcal{C}$.

After performing this maneuver, ${\bf F}$ can subsequently be treated as a defect between the trivial domain wall in $\bf D$ and the composite domain wall. Mathematically, this statement corresponds to the fact that every $(\mathcal{C},\mathcal{D})$-bimodule functor, $F:\mathcal{M}\to\mathcal{N}$, induces a $(\mathcal{D},\mathcal{D})$-bimodule functor $\mathcal{D}\to \mathcal{M}^* \boxtimes_{\mathcal{C}} \mathcal{N}$ by Frobenius reciprocity \cite[Corollary 3.22]{greenough2010monoidal} \footnote{We use the following chain of equivalences
$ \Fun_{\CC|\CD}(\CM,\CN)\xrightarrow{\sim} \Fun_{\CC|\CD}(\CM\boxtimes_{\CD}\CD,\CN)\xrightarrow{\sim}   
    \Fun_{\CD|\CD}(\CD,\Fun_{\CC|\CD}(\CM,\CN))\xrightarrow{\sim} \Fun_{\CD|\CD}(\CD,\CM^*\boxtimes_{\CC}\CN)$}. By abuse of notation, we also call this induced functor $F$.

The final step of the pinching trick is to observe that the information contained in the half-annular region around $\bf F$ is exactly the information characterizing which simple domain wall the composite domain wall fuses to. The possible fusion channels for the composite domain wall $\mathcal{M}^* \boxtimes_{\mathcal{C}} \mathcal{N}$ correspond to indecomposable sub-$(\mathcal{D},\mathcal{D})$-bimodules of $\mathcal{M}^* \boxtimes_{\mathcal{C}} \mathcal{N}$ \cite{huston2023composing}. Mathematically, this means that we need to identify a distinguished indecomposable bimodule subcategory of $\mathcal{M}^* \boxtimes_{\mathcal{C}} \mathcal{N}$ corresponding to the functor $F$.

To identify the correct indecomposable sub-bimodule, let us recall the general algebraic theory of domain walls. The trivial defect on the composite domain wall, $\mathcal{M}^* \boxtimes_{\mathcal{C}} \mathcal{N}$, can be decomposed using a resolution of the identity, where each term is a projection onto an indecomposable bimodule subcategory followed by the inclusion of the subcategory. This logic is depicted using string diagrams as follows:

\begin{equation}
\label{resolution-of-id}
\begin{tikzpicture}
	\begin{pgfonlayer}{nodelayer}
		\node [style=none] (0) at (0.5, 0.725) {};
		\node [style=none] (1) at (0.5, -0.675) {};
		\node [style=none] (2) at (1.5, 0.725) {};
		\node [style=none] (3) at (1.5, -0.675) {};
		\node [style=none] (4) at (0.5, -1.125) {$\mathcal{N}^*$};
		\node [style=none] (5) at (1.5, -1.125) {$\mathcal{M}$};
		\node [style=none] (6) at (0.5, 1.175) {$\mathcal{N}^*$};
		\node [style=none] (7) at (1.5, 1.175) {$\mathcal{M}$};
	\end{pgfonlayer}
	\begin{pgfonlayer}{edgelayer}
		\draw (0.center) to (1.center);
		\draw (2.center) to (3.center);
	\end{pgfonlayer}
\end{tikzpicture}
=\sum_{\mathfrak{n}\subseteq \mathcal{N}^*\boxtimes_{\mathcal{C}}\mathcal{M}}\begin{tikzpicture}
	\begin{pgfonlayer}{nodelayer}
		\node [style=none] (0) at (-1, -0.5) {};
		\node [style=none] (1) at (-1, -0.7) {};
		\node [style=none] (2) at (0, -0.5) {};
		\node [style=none] (3) at (0, -0.7) {};
		\node [style=none] (4) at (-1, -1.15) {$\mathcal{N}^*$};
		\node [style=none] (5) at (0, -1.15) {$\mathcal{M}$};
		\node [style=none] (6) at (-1, 1.2) {$\mathcal{N}^*$};
		\node [style=none] (7) at (0, 1.2) {$\mathcal{M}$};
		\node [style=none] (8) at (-1, 0.75) {};
		\node [style=none] (9) at (0, 0.75) {};
		\node [style=none] (10) at (-1, 0.55) {};
		\node [style=none] (11) at (0, 0.55) {};
		\node [style=none] (12) at (-0.5, 0.3) {};
		\node [style=none] (13) at (-0.5, -0.25) {};
		\node [style=none] (14) at (-0.15, 0) {$\mathfrak{n}$};
	\end{pgfonlayer}
	\begin{pgfonlayer}{edgelayer}
		\draw (0.center) to (1.center);
		\draw (2.center) to (3.center);
		\draw (9.center) to (11.center);
		\draw (8.center) to (10.center);
		\draw (13.center) to (12.center);
		\draw (12.center) to (10.center);
		\draw (12.center) to (11.center);
		\draw (13.center) to (0.center);
		\draw (13.center) to (2.center);
	\end{pgfonlayer}
\end{tikzpicture}
~.
\end{equation}

Measuring the type of a composite domain wall corresponds to measuring with respect to the resolution in \eqref{resolution-of-id}. In our present situation, the only non-vanishing term in the decomposition is the subcategory in the image of $F$. As such, the bimodule subcategory associated with $F$ must be its image:

\begin{equation}
\label{n-type}
\mathfrak{n}_F=\text{Im}(F:\mathcal{D}\to \mathcal{M}^* \boxtimes_{\mathcal{C}} \mathcal{N})~.
\end{equation}

We call $\mathfrak{n}_F$ the \textit{$N$-type parton associated with $F$}. It is straightforward to verify that $\mathfrak{n}_F$ is an indecomposable bimodule subcategory using the fact that $F:\mathcal{C}\to\mathcal{D}$ is an irreducible bimodule functor. As such, we see that $\mathfrak{n}_F$ is a valid choice of fusion channel.

Arguing in an analogous fashion, we see that $U$-type parton sectors contain the information of an indecomposable $(\mathcal{C},\mathcal{C})$-bimodule subcategory of $\mathcal{N}\boxtimes_{\mathcal{D}}\mathcal{M}^*$. Associating a functor, $F:\mathcal{C}\to \mathcal{N}\boxtimes_{\mathcal{D}}\mathcal{M}^*$, with $F:\mathcal{M}\to\mathcal{N}$ via Frobenius reciprocity, the $U$-type parton corresponding to $F$ is 

\begin{equation}
\label{u-type}
\mathfrak{u}_F=\text{Im}(F:\mathcal{C}\to \mathcal{N}\boxtimes_{\mathcal{D}}\mathcal{M}^*)~.
\end{equation}

We summarize our proposal in table \ref{parton-classification}, maintaining our previously established conventions.

\begin{table}[h]
\begin{tabular}{ |p{3.7cm}||p{3.7cm}|  }

 \hline
 Entanglement bootstrap& Category theory\\
 \hline
$N$-type parton sectors   & Indecomposable $(\mathcal{D},\mathcal{D})$-bimodule subcategories of $\mathcal{M}^* \boxtimes_{\mathcal{C}} \mathcal{N}$\\
 $U$-type parton sectors & Indecomposable $(\mathcal{C},\mathcal{C})$-bimodule subcategories of $\mathcal{N}\boxtimes_{\mathcal{D}}\mathcal{M}^*$ \\
 $N$-type parton sector associated with a defect, F& $\, \mathfrak{n}_F\subseteq \mathcal{M}^* \boxtimes_{\mathcal{C}} \mathcal{N}$ (\ref{n-type}) \\
$U$-type parton sector associated with a defect, F  &  
 $\,\mathfrak{u}_F\subseteq \mathcal{N}\boxtimes_{\mathcal{D}}\mathcal{M}^*$ (\ref{u-type}) \\
 \hline
\end{tabular}
\caption{Our proposed algebraic theory of partons.}
\label{parton-classification}
\end{table}

\section{Grading by partons}\label{grading}

By construction, the space of domain wall defects can be decomposed into different parton sectors. Mathematically, this decomposition endows the fusion category of domain wall excitations with a structure reminiscent of a grading. In this section we develop this picture further and demonstrate that the neutral component of this grading can be endowed with the structure of a modular tensor category.

As before, let $\mathcal{C}$, $\mathcal{D}$ be unitary fusion categories, and let $\mathcal{M}$ be an indecomposable $(\mathcal{C},\mathcal{D})$-bimodule. The fusion category, $\mathcal{E}:=\Fun_{\mathcal{C}|\mathcal{D}}(\mathcal{M},\mathcal{M})$, of bimodule endofunctors of $\mathcal{M}$ describes the excitations on the domain wall labeled by $\mathcal{M}$.

We introduce some notation, roughly following \cite{ShiKim2021}. Define $\mathcal{L}_N$ ($\mathcal{L}_U$) to be the set of indecomposable $(\mathcal{D},\mathcal{D})$-bimodule subcategories ($(\mathcal{C},\mathcal{C})$-bimodule subcategories) of $\mathcal{M}^*\boxtimes_{\mathcal{C}}\mathcal{M}$ ($\mathcal{M}\boxtimes_{\mathcal{D}}\mathcal{M}^*$), and let $\mathcal{L}_{O}^{[\mathfrak{n},\mathfrak{u}]}$ be the set of irreducible bimodule functors, $F:\mathcal{M}\to\mathcal{M}$, such that $\mathfrak{n}_F=\mathfrak{n}\in\mathcal{L}_N$ and $\mathfrak{u}_F=\mathfrak{u}\in\mathfrak{L}_U$. We denote by ${\bf 1}\in \mathcal{L}_N$ (${\bf 1}\in \mathcal{L}_U$) the bimodule category generated by the identity functor in $\mathcal{M}^*\boxtimes_{\mathcal{C}}\mathcal{M}\simeq \Fun_{\mathcal{C}}(\mathcal{M},\mathcal{M})$ ($\mathcal{M}\boxtimes_{\mathcal{D}}\mathcal{M}^*\simeq \Fun_{\mathcal{D}}(\mathcal{M},\mathcal{M})$).

 For each $\mathfrak{n}\in \mathcal{L}_N$, $\mathfrak{u}\in \mathcal{L}_U$, let $\mathcal{E}^{[\mathfrak{n},\mathfrak{u}]}\subseteq \mathcal{E}$ be the full subcategory additively generated by the simple objects in $\mathcal{L}_O^{[\mathfrak{n},\mathfrak{u}]}$. Define 

\begin{equation}
 \mathcal{E}^{[\bullet,{\bf1}]}:=\bigoplus_{\mathfrak{n}\in \mathcal{L}_N}\mathcal{E}^{[\mathfrak{n},{\bf 1}]}~,\quad\mathcal{E}^{[{\bf 1},\bullet]}:=\bigoplus_{\mathfrak{u}\in \mathcal{L}_U}\mathcal{E}^{[{\bf 1},\mathfrak{u}]}~.
 \end{equation}

Acting on the left and right respectively, there are induced functors from $\mathcal{Z}(\mathcal{C})$ and $\mathcal{Z}(\overline{\mathcal{D}})$ to $\mathcal{E}$. Unpacking definitions, one can see that $ \mathcal{E}^{[\bullet,{\bf1}]}$ and $\mathcal{E}^{[{\bf 1},\bullet]}$ are equal to the images of $\mathcal{Z}(\mathcal{C})$ and $\mathcal{Z}(\overline{\mathcal{D}})$. As such, $\mathcal{E}^{[\bullet,{\bf1}]}$, $\mathcal{E}^{[{\bf 1},\bullet]}$, and $\mathcal{E}^{[{\bf 1},{\bf 1}]}=\mathcal{E}^{[\bullet,{\bf1}]}\cap \mathcal{E}^{[{\bf 1},\bullet]}$ are all fusion subcategories of $\mathcal{E}$. The following decomposition acts similarly to a grading on $\mathcal{E}$:
\begin{equation}
    \mathcal{E}=\bigoplus_{\substack{\mathfrak{n}\in \mathcal{L}_N \\ \mathfrak{u}\in\mathcal{L}_U}}\mathcal{E}^{[\mathfrak{n},\mathfrak{u}]}~.
\end{equation}

We now turn our attention to the neutral component,  $\mathcal{E}^{[{\bf 1},{\bf 1}]}\subseteq  \mathcal{E}$. A priori, this component is only endowed with the structure of a fusion subcategory. However, we will argue that $\mathcal{E}^{[{\bf 1},{\bf 1}]}$ can be naturally equipped with a non-degenerate braiding, making it a uMTC. The physical interpretation of this braiding is as follows. The $\mathcal{E}^{[{\bf 1},{\bf 1}]}$ sectors, by construction, are exactly the sectors which can be pulled into the bulk in either direction. The braiding is defined by pulling one of the excitations into the bulk one direction, bringing the other other excitation into the bulk the other direction, braiding in the bulk, and then going back to the domain wall. Proposition \ref{braiding-well-defined} explains mathematically why this braiding does not depend on the way we pull the domain wall sectors into the bulk.

Let $A,B\in \mathcal{E}^{[{\bf 1},{\bf 1}]}$ be objects. Since $A\in \mathcal{E}^{[\bullet ,\bf 1]}$, we know that there exists some object, $C\in \mathcal{Z}(\mathcal{C})$, and an inclusion $A\hookrightarrow{} F_C$ where $F_C:\mathcal{M}\to \mathcal{M}$ is the functor $M\mapsto C\otimes M$. This discussion is a re-statement of the fact that $\mathcal{E}^{[\bullet ,\bf 1]}$ is the image of $\mathcal{Z}(\mathcal{C})$. Similarly, since $B\in \mathcal{E}^{[\bf 1, \bullet]}$, we know that there exists some object $D\in\mathcal{Z}(\mathcal{D})$ and an inclusion $B\hookrightarrow{}F_D$ where $F_D:\mathcal{M}\to\mathcal{M}$ is the functor $M\mapsto M\otimes D$. We can now define a map $\beta_{A,B}:A\otimes B\to B\otimes A$ as the unique morphism making the following diagram commute:

\begin{equation}
\label{braiding-def}
\begin{tikzcd}
	{A\otimes B} & {B\otimes A} \\
	{F_C\otimes F_D} & {F_D\otimes F_C}
	\arrow["{\beta_{A,B}}", from=1-1, to=1-2]
	\arrow[from=1-1, to=2-1]
	\arrow["{b_{C,D}}", from=2-1, to=2-2]
	\arrow[from=2-2, to=1-2]
\end{tikzcd}    
\end{equation}

Here, $b_{C,D}$ is the bimodule associativity map associated with $\mathcal{M}$,

\begin{equation}
    b_{C,D}:C\otimes (\bullet \otimes D)\xrightarrow{} (C\otimes \bullet )\otimes D~.
\end{equation}

\begin{Proposition}\label{braiding-well-defined} The maps $\beta_{A,B}$ defined in (\ref{braiding-def}) are independent of the choices of $C\in\mathcal{Z}(\mathcal{C})$ and $D\in\mathcal{Z}(\mathcal{D})$, and they endow the category $\mathcal{E}^{[\bf 1,\bf 1]}$ with the structure of a braided unitary fusion category.
\end{Proposition}
\begin{proof} To show that $\beta_{A,B}$ is independent of choices, we must show that for $C,C'\in \mathcal{Z}(\mathcal{C})$ and $D,D'\in \mathcal{Z}(\mathcal{D})$ the maps $\beta_{A,B}$ associated with $C,D$ and $C',D'$ are the same. Consider the bimodule functors $f_A:F_C\to F_{C'}$ and $f_B:F_D\to F_{D'}$ defined by the compositions

\begin{equation}
\begin{tikzcd}
	{F_C} & {F_{C'}} && {F_D} & {F_{D'}} \\
	& A &&& B
	\arrow["{f_A}", from=1-1, to=1-2]
	\arrow[from=1-1, to=2-2]
	\arrow["{f_B}", from=1-4, to=1-5]
	\arrow[from=1-4, to=2-5]
	\arrow[hook, from=2-2, to=1-2]
	\arrow[hook, from=2-5, to=1-5]
\end{tikzcd}
\end{equation}

Unpacking definitions, we find it suffices to demonstrate commutativity of the square

\begin{equation}
\begin{tikzcd}
	{F_C\otimes F_D} & {F_{D}\otimes F_C} \\
	{F_{C'}\otimes F_{D'}} & {F_{D'}\otimes F_{C'}}
	\arrow["{b_{C,D}}", from=1-1, to=1-2]
	\arrow["{f_A\otimes f_B}"', from=1-1, to=2-1]
	\arrow["{f_B\otimes f_A}", from=1-2, to=2-2]
	\arrow["{b_{C',D'}}"', from=2-1, to=2-2]
\end{tikzcd}
\end{equation}

Breaking up the square into a first application of $f_A$ and a second application of $f_B$, the commutativity comes from first applying the compatibility of $f_A$ with bimodule associativity and then applying the compatibility of $f_B$ with bimodule associativity.

The fact that $\beta_{A,B}$ satisfies the axioms of a braiding is a straightforward computation using the compatibility of the bimodule associativity with the bimodule actions on each side.
\end{proof}

We now give a physical argument for why the braiding on $\mathcal{E}^{[1,1]}$ should be non-degenerate. It is a well-known fact that every domain wall can be decomposed by first performing anyon condensation, then acting with an invertible domain wall, and finally performing reverse anyon condensation, as shown in Fig.~\ref{fig:sandwich}. The key observation is that the domain wall excitations that can be pulled into the bulk on either side are exactly the domain wall excitations that can live in the middle condensed phase.  As such, $\mathcal{E}^{[{\bf 1},{\bf 1}]}$ describes anyons in the condensed phase and should therefore have non-degenerate braiding. This argument can be made mathematically precise by the following result:

\begin{figure}
    \centering   \includegraphics[width=0.65\linewidth]{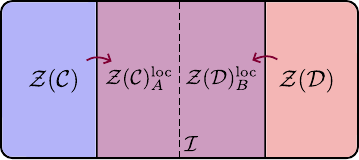}
    \caption{Every gapped domain wall between two string-net models can be obtained by juxtaposing two condensation domain walls and an invertible domain wall. This figure is similar to \cite[Fig. 7]{Huston2022}.}
    \label{fig:sandwich}
\end{figure}

\begin{Proposition}\label{non-degeneracy-theorem} Suppose that there is a $(\mathcal{C},\mathcal{D})$-bimodule equivalence $\mathcal{M}\cong\mathcal{C}_A\boxtimes_{\mathcal{C}_A^{\text{loc}}} \mathcal{I}\boxtimes_{\mathcal{D}_B^{\text{loc}}} \mathcal{D}_B$ where $A\in \mathcal{C}$, $B\in \mathcal{D}$ are separable commutative algebras. Then, there are equivalences of unitary braided fusion categories

\begin{equation}
    \mathcal{E}^{[{\bf 1}, {\bf 1}]}\cong \mathcal{Z}(\mathcal{C})^{\mathrm{loc}}_{A}\cong \mathcal{Z}(\mathcal{D})^{\mathrm{loc}}_B~.
\end{equation}

In particular, the braiding on $\mathcal{E}^{[\bf 1,\bf 1]}$ is always non-degenerate.
\end{Proposition}

\begin{proof} Consider the following monoidal functor, obtained by using the naturality of the relative Deligne tensor product:

\begin{equation}
\label{tensor-functor}
\Fun_{\mathcal{X}|\mathcal{Y}}(\mathcal{I},\mathcal{I})\xrightarrow{}\Fun_{\mathcal{C}|\mathcal{D}}(\mathcal{M},\mathcal{M})~,
\end{equation}
where $Z(\mathcal{X})=\mathcal{Z}(\mathcal{C})^{\mathrm{loc}}_{A}$ and $Z(\mathcal{Y})=\mathcal{Z}(\mathcal{D})^{\mathrm{loc}}_{B}$. We now argue that every object in the image of the map (\ref{tensor-functor}) lives in $\mathcal{E}^{[\bf 1,\bf 1]}$. Every functor in $\Fun_{\CX|\CY}(\mathcal{I},\mathcal{I})$ comes from tensoring on the left with an element of $\mathcal{Z}(\mathcal{C})_A^{\text{loc}}$, since $\mathcal{I}$ is invertible. Pushing forward to $\Fun_{\mathcal{C}|\mathcal{D}}(\mathcal{M},\mathcal{M})$ and using the universal property of the relative Deligne tensor product, we find that the functors in the image must all come from acting just on the factor of $\mathcal{C}_A$. Now, since every $(\mathcal{C},\CX)$-bimodule endofunctor of $\mathcal{C}_A$ comes from tensoring on the left with an element of $\mathcal{Z}(\mathcal{C})$, we conclude that every object in the image of the map (\ref{tensor-functor}) lives in $\mathcal{E}^{[\bf 1,\bullet ]}$. By a completely symmetric argument, we thus find that every functor is in  $\mathcal{E}^{[\bf 1,\bf 1]}=\mathcal{E}^{[\bullet, \bf 1]}\cap \mathcal{E}^{[\bf 1,\bullet ]}$. Playing this argument in reverse, we conclude there is an equivalence of unitary fusion categories

\begin{equation}\label{E-to-F}
    \mathcal{E}^{[\bf 1,\bf 1]}\cong \Fun_{\CX|\CY}(\mathcal{I},\mathcal{I})~.
\end{equation}

Now, using Frobenius reciprocity and the invertibility of $\mathcal{I}$, we get equivalences of fusion categories

\begin{align*}
 \Fun_{\CX|\CY}(\mathcal{I},\mathcal{I})&\cong \Fun_{\CX|\CX}(\CX, \mathcal{I}\boxtimes_{\CY}\mathcal{I^*})\\
 &\cong \Fun_{\CX|\CX}(\CX, \CX)\cong \mathcal{Z}(\CC)_A^{\rm loc}~.
\end{align*}

Tracking through the braiding and combining with (\ref{E-to-F}), we arrive at the equivalence of $\mathcal{E}^{[\bf 1, \bf 1]}$ and $\mathcal{Z}(\mathcal{C})_A^{\text{loc}}$. Using Frobenius reciprocity in the other direction, we obtain the equivalence of $\mathcal{E}^{[\bf 1, \bf 1]}$ and $\mathcal{Z}(\mathcal{D})_B^{\text{loc}}$.

The fact that every bimodule category, $\mathcal{M}$, has the form given in Proposition \ref{non-degeneracy-theorem} comes from applying \cite[Proposition 3.6]{davydov2013structure}. As such, the category $\mathcal{E}^{[\bf 1, \bf 1]}$ is always equivalent to the Drinfeld center of some unitary fusion category and is therefore non-degenerate.
\end{proof}

\section{Quantum dimensions}\label{sec:d_n}

In the entanglement bootstrap program, one can assign quantum dimensions to many types of sectors by looking at their relative von Neumann entropies compared with the vacuum sector. As a particular application of this idea, one of the present authors in collaboration with Kim introduced the notion of the quantum dimension of a parton sector \cite{ShiKim2021}. Since this definition of quantum dimension requires the existence of a vacuum sector, partons only naively have quantum dimensions when they live at the interface between two domain walls of the same type. In this case, the vacuum parton sector (denoted $\bf 1$) is defined to be the parton sector associated with the trivial domain wall defect (i.e., no defect at all). Nonetheless, in Appendix~\ref{app:dn-M-N} we put forward a broader definition, which is applicable without reference to a vacuum.

Passing through our algebraic correspondence, the first-principles considerations in \cite{ShiKim2021} predict a \lq\lq quantum dimension" invariant, $d_{\mathfrak{n}}$, associated with indecomposable bimodule subcategories $\mathfrak{n}\subseteq \mathcal{M}^* \boxtimes_{\mathcal{C}} \mathcal{M}$, and a quantum dimension, $d_{\mathfrak{u}}$,  associated with indecomposable bimodule subcategories $\mathfrak{u}\subseteq \mathcal{M}\boxtimes_{\mathcal{D}}\mathcal{M}^*$. In fact, the internal consistency conditions between the quantum dimensions of parton sectors and quantum dimensions of domain wall excitations derived in \cite{ShiKim2021} are enough to uniquely specify these invariants. Namely, \cite[eq. (56)]{ShiKim2021} reads
 \begin{equation}
 \label{dn-times-du}
 d^2_{\mathfrak{n}}d^2_{\mathfrak{u}}=\frac{\sum_{F\in \mathcal{L}_O^{[\mathfrak{n},\mathfrak{u}]}}d_{F}^2}{\sum_{F\in \mathcal{L}_O^{[{\bf1},{\bf1}]}} d_F^2}~,
 \end{equation}
where $\mathcal{L}_{N}$, $\mathcal{L}_{U}$, and $\mathcal{L}_{O}^{[\mathfrak{n},\mathfrak{u}]}$ are as before. Here, the quantum dimension $d_F$ can be defined using the unitary fusion category structure on the space $\Fun_{\mathcal{C}|\mathcal{D}}(\mathcal{M},\mathcal{M})$ of domain wall excitations.
 
 Applying this formula first with $\mathfrak{u}={\bf 1}$ and then with $\mathfrak{n}={\bf 1}$ we arrive at expressions for $d_{\mathfrak{n}}$ and $d_{\mathfrak{u}}$ entirely in terms of the quantum dimensions of domain wall excitations
 \begin{equation}
 \label{dn-and-du-def}
 d^2_{\mathfrak{n}}=\frac{\sum_{F\in \mathcal{L}_O^{[\mathfrak{n},{\bf 1}]}}d_{F}^2}{\sum_{F\in \mathcal{L}_O^{[{\bf1},{\bf1}]}} d_F^2},\quad d^2_{\mathfrak{u}}=\frac{\sum_{F\in \mathcal{L}_O^{[{\bf 1},\mathfrak{u}]}}d_{F}^2}{\sum_{F\in \mathcal{L}_O^{[{\bf1},{\bf1}]}} d_F^2}~.
 \end{equation}
Clearly, (\ref{dn-and-du-def}) is therefore a valid mathematical definition of $d_{\mathfrak{n}}$ and $d_{\mathfrak{u}}$.

In this way, \eqref{dn-times-du} becomes a non-trivial consistency check of our category theoretic framework. It provides a concrete conjecture about fusion categories predicted from first principles. We state our discussion formally as a theorem and give a proof:

 \begin{theorem}\label{thm:d_n-M-M} 
 Equation \eqref{dn-times-du} is valid for all pairs of unitary fusion categories $\mathcal{C}$, $\mathcal{D}$ and for all indecomposable $(\mathcal{C},\mathcal{D})$-bimodule categories $\mathcal{M}$.
 \end{theorem}
 \begin{proof} Let $\mathcal{E}=\Fun_{\mathcal{C}|\mathcal{D}}(\mathcal{M},\mathcal{M})$ be as before. Define $R^{[\mathfrak{n},\mathfrak{u}]}=\sum_{F\in \mathcal{L}_O^{[\mathfrak{n},\mathfrak{u}]}}d_F[F]$  living in the Grothendieck ring $K_0(\mathcal{E})\otimes \mathbb{R}$ of $\mathcal{E}$, and let $R^{[\bullet, {\bf 1}]}$, $R^{[{\bf 1},\bullet]}$, and $R$ be defined in terms appropriate linear combinations of simple objects in $R^{[\mathfrak{n},\mathfrak{u}]}$.

By \cite[Prop. 2.34]{drinfeld2010braided} the functor $\mathcal{Z}(\mathcal{C}\boxtimes \overline{\mathcal{D}})\to \mathcal{E}$ is surjective, and thus $\mathcal{E}$ has no proper tensor subcategories containing both $\mathcal{E}^{[\bullet,{\bf1}]}$ and $\mathcal{E}^{[{\bf1},\bullet]}$. Additionally, for any $X\in \mathcal{E}^{[\bullet,{\bf1}]}$ and $Y\in \mathcal{E}^{[{\bf1},\bullet]}$ we have $X\otimes Y\cong Y\otimes X$ since $X$, $Y$ act by module actions on opposite sides. As such, we are in the situation to apply \cite[Lemma 3.38]{drinfeld2010braided}, which reads
\begin{equation}
\label{regular-element-lemma}
R^{[\bullet ,{\bf 1}]}\otimes R^{[{\bf 1},\bullet ]}=\dim(\mathcal{E}^{[{\bf 1},{\bf 1}]})R{~,}
\end{equation}
where \lq\lq$\dim$" is the global quantum dimension. Since the tensor product restricts to a map $\otimes: \mathcal{E}^{[\mathfrak{n},{\bf 1}]}\boxtimes \mathcal{E}^{[{\bf 1},\mathfrak{u}]}\to \mathcal{E}^{[\mathfrak{n},\mathfrak{u}]}$, we can compare the terms in the span of $\mathcal{L}_O^{[\mathfrak{n},\mathfrak{u}]}$ on both sides of eq. (\ref{regular-element-lemma}) to get
\begin{equation}
\label{regular-element-lemma-2}
R^{[\mathfrak{n} ,{\bf 1}]}\otimes R^{[{\bf 1},\mathfrak{u} ]}=\dim(\mathcal{E}^{[{\bf 1},{\bf 1}]})R^{[\mathfrak{n},\mathfrak{u}]}.
\end{equation}
Taking quantum dimensions and performing algebraic manipulations we obtain (\ref{dn-times-du}).

\end{proof}

\section{A Spacetime Picture and Generalized Symmetries}\label{GenSym}
In this section, we aim to give a spacetime-covariant perspective on the previous sections. Instead of thinking in terms of bulk Levin-Wen models, we place our $\calC$ and $\CD$ fusion categories on a two-spacetime-dimensional boundary, $\Sigma$, of a three-spacetime-dimensional bulk three-manifold, $X$ (i.e., $\partial X=\Sigma$).

The bulk topological orders on $X$ are described by Drinfeld centers, $\CZ(\calC)$ and $\CZ(\CD)$, separated by a two-spacetime-dimensional domain wall. On this domain wall, we place a simple one-spacetime-dimensional topological excitation, $F$, that divides the domain wall into surfaces $\CM$ and $\CN$. $\CM$ and $\CN$ intersect $\Sigma$ as one-spacetime-dimensional interfaces, $\CM_{\Sigma}$ and $\CN_{\Sigma}$, separating $\calC$ and $\CD$, while $F$ maps to a junction, $F_{\Sigma}$, separating $\CM_{\Sigma}$ and $\CN_{\Sigma}$ (see Fig. \ref{ST1}).

\begin{figure}
\includegraphics[width=7.3cm]{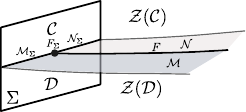}
\caption{A spacetime covariant version of the initial configuration that gives rise to partons. Here $\CN$ and $\CM$ are topological domain walls between TQFTs, $\CZ(\calC)$ and $\CZ(\CD)$, on a 3-manifold, $X$, and $F$ is a topological interface between these surfaces. The $\calC$ and $\CD$ generalized symmetries live on a boundary, $\Sigma$, where they act on $\CM_{\Sigma}$, $\CN_{\Sigma}$, and $F_{\Sigma}$. This action associates $(\calC,\CD)$ symmetry multiplets with $\CM_{\Sigma}$ and $\CN_{\Sigma}$ corresponding to $\CM$ and $\CN$ as well as an irrep of the generalized tube algebra, ${\rm Tube}(\CM|\CN)$, with $F_{\Sigma}$ and $F$ \cite{Choi:2024tri,Bhardwaj:2024igy}.}
\label{ST1}
\end{figure}

From the perspective of generalized symmetries, $\calC$ and $\CD$ are thought of as (non-)invertible fusion category symmetries on $\Sigma$. Then, using the physical picture in \cite{Choi:2024tri,Bhardwaj:2024igy}, we can think of $\CM_{\Sigma}$ and $\CN_{\Sigma}$ as belonging to irreducible $(\calC,\CD)$ multiplets of boundary conditions for the surfaces $\CM$ and $\CN$, where the multiplets in question are generated by the action of the $\calC$ and $\CD$ symmetries. More mathematically, these symmetry multiplets are $(\calC,\CD)$-bimodule categories. As we have discussed previously, the bulk $\CM$ and $\CN$ surfaces can also be thought of as $(\calC,\CD)$-bimodule categories. These bimodule categories correspond to the $(\CC,\CD)$ symmetry multiplets that $\CM_{\Sigma}$ and $\CN_{\Sigma}$ transform under. In this language, $F_{\Sigma}$ and $F$ are associated with an irreducible representation of the generalized tube algebra, ${\rm Tube}(\CM|\CN)$ \cite{Choi:2024tri}, under the action of $\calC$ and $\CD$. 

Now we consider the pinching trick in this language. To that end, the $N$-type parton sectors are in one-to-one correspondence with the irreducible $(\CD,\CD)$ symmetry multiplets that arise when we fold the domain wall in such a way as to shrink the $\CC$ region. In other words, an $N$-type parton, $\kappa$, corresponds to a simple surface that arises in the bulk fusion of $\CM^*$ and $\CN$.

\begin{figure}
\includegraphics[width=8cm]{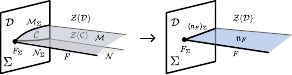}
\caption{We pinch the domain wall separating $\CZ(\calC)$ and $\CZ(\CD)$. In so doing, $F$, becomes a topological boundary condition for a surface, $\mathfrak{n}_F$, in the product of $\CM^*$ with $\CN$. This is the parton associated with $F$.}
\label{PinchST}
\end{figure}

Notice that once we fuse $\CM^*$ and $\CN$, $F$ becomes a topological boundary condition for one of the simple $\kappa$ surfaces. In this picture, the corresponding parton, $\mathfrak{n}_F$, is precisely the surface we associate with $F$ (see Fig. \ref{PinchST}). We can think of the parton quantum number as associating with $F_{\Sigma}$ and $F$ an irreducible representation of ${\rm Tube}(\mathfrak{n}_{F})$ given an irreducible representation of ${\rm Tube}(\CM|\CN)$. We may repeat the above discussion for the $U$-type parton sectors by instead folding the domain wall so that the $\CD$ region shrinks.

Given that we associate partons with surfaces, it may be surprising that we can assign quantum dimensions to them as in \eqref{dn-and-du-def}. However, as the examples in the appendix make clear, these quantum dimensions do not generally correspond to the naive quantum dimensions we read off from the 2-categorical fusion rules for surfaces in a particular topological phase. In fact, such quantities are typically ill-defined in quantum field theory (QFT) \footnote{To see one source of ambiguity, note that we can add an Euler counterterm living on the surface. This counterterm rescales the surface by $s^{\chi(\Sigma)}\in\mathbb{R}$ and therefore affects the fusion rules and corresponding quantum dimensions. A more general ambiguity corresponds to the fact that, in $D+1$ dimensional QFT, one expects $d$-dimensional defects to have fusion coefficients valued in $d$-dimensional TQFTs. See \cite{Roumpedakis2023} for a recent discussion.}. The main point is that our definition of the parton quantum dimension is in terms of one-spacetime-dimensional objects (the domain wall excitations) which do have well-defined quantum dimension. Moreover, our quantum dimensions are not properties of surfaces in a particular topological phase but rather are properties of surfaces arising from pairs of topological phases (separated by a gapped domain wall).

Finally, let us note that the above discussion is closely related to the construction of the SymTFT associated with two, in general different, two-spacetime-dimensional QFTs, $Q_1$ and $Q_2$, with fusion category symmetries $\calC$ and $\CD$ separated by a one-spacetime-dimensional interface (see \cite{Choi:2024tri,Bhardwaj:2024igy} and references therein). The main differences with respect to our above construction are that in the SymTFT we have a second boundary so that $X=I\times\Sigma$ (where $I\cong[0,1]$ is an interval) with both boundaries isomorphic to $\Sigma$ (we denote them as $\Sigma_{L,R}$), and $Q_{1,2}$ are not necessarily topological.

\begin{figure}
\includegraphics[width=7.8cm]{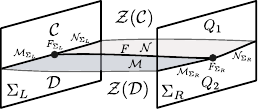}
\caption{The SymTFT embedding of Fig. \ref{ST1}. We include an additional right boundary, $\Sigma_R$, supporting (in general) non-topological QFTs, $Q_1$ and $Q_2$. Here the left boundary, $\Sigma_L$, supports the $\calC$ and $\CD$ generalized symmetries. The bulk $\CM$ and $\CN$ surfaces and $F$ interface end topologically on $\Sigma_L$ and (in general) non-topologically on $\Sigma_R$.}
\label{SymTFT}
\end{figure}

The SymTFT idea is that we can take $Q_1$ and $Q_2$ separated by a (not necessarily topological) interface and \lq\lq blow the theory up" into a theory on the three-manifold $X$, where the symmetries of $Q_{1,2}$ live on the left boundary, $\Sigma_L$, and the dynamical content of these QFTs live on $\Sigma_R$. The bulk is topological, and the generalized charges of the theory correspond to extended topological operators stretching between the boundaries (see Fig. \ref{SymTFT}). Then, applying our discussion above, operators that live on the interface between $Q_1$ and $Q_2$ carry parton quantum numbers. Moreover, we see that we can associate $N$-type parton quantum numbers (or, equivalently, irreps of ${\rm Tube}(\mathfrak{n}_F)$ that descend from irreps of ${\rm Tube}(\CM|\CN)$) with twisted-sector states arising from the pinching of the interfaces between $Q_1$ and $Q_2$ in such a way that the region supporting $Q_1$ shrinks (see Fig. \ref{Twist}). Similar comments apply to pinching in the opposite direction and $U$-type partons.

\begin{figure}
\includegraphics[width=8cm]{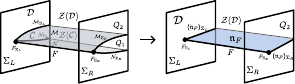}
\caption{Performing the pinching trick in the SymTFT associates a parton, $\mathfrak{n}_F$, with a (generally non-topological) twisted sector operator in QFT $Q_2$ gotten from dimensionally reducing the above configuration to two dimensions.}
\label{Twist}
\end{figure}

We summarize our discussion in Table \ref{parton-Gsymm}.

\begin{table}[h]
\begin{tabular}{ |p{3.8cm}||p{4.2cm}|}

\hline
Entanglement bootstrap& Spacetime picture\\
 \hline
 $N$-type parton sectors   & Simple topological surfaces in fusion of $\mathcal{M}^*$ and $\mathcal{N}$ \\
 \hline 
 $U$-type parton sectors & Simple topological surfaces in fusion of $\mathcal{N}$ and $\mathcal{M}^*$ \\
 \hline
 $N$-type parton sector associated with a defect, $F$&  Simple topological surface in fusion of $\mathcal{M}^*$ and $\mathcal{N}$ associated with $F$ topological boundary condition\\
 \hline
 $U$-type parton sector associated with a defect, $F$&  Simple topological surface in fusion of $\mathcal{N}$ and $\mathcal{M}^*$ associated with $F$ topological boundary condition\\
 \hline
\end{tabular}
\caption{A spacetime-covariant description of entanglement bootstrap concepts. In this section, we describe how to map these concepts to generalized symmetry multiplets and how to associate $F$ with an element of a (generalized) tube algebra.}
\label{parton-Gsymm}
\end{table}

\section{Discussion}

In this note, we have argued for an algebraic interpretation of entanglement bootstrap parton sectors \cite{ShiKim2021} within Kitaev and Kong's categorical framework of string-nets and bimodule categories. We then described our dictionary from a spacetime covariant perspective and also made contact with recent work on SymTFTs and generalized symmetries. In a set of appendices, we describe certain extensions of the discussion in the main text, consider various illustrative examples, and also give a more detailed exposition of the entanglement bootstrap while obtaining some new results.

There are several natural extensions of the techniques discussed in this note. One that we touch upon in a particular set of examples in Appendix~\ref{WittNT}, is to the case of phases that do not admit gapped boundaries, and are thus not in the same phase as Levin-Wen string nets~\cite{levin2005string,kim2024classifying} (theories that we refer to as corresponding to {\em Witt non-trivial} phases~\cite{Levin2013,Ng2020}).  While the Kitaev-Kong framework cannot be directly used in this situation, the techniques from this note are still applicable either via judicious use of the folding trick (e.g., see \cite{fuchs2013bicategories}) or via combining our methods with the formalism in \cite{huston2023composing}.

Another natural extension is to consider other entanglement bootstrap sectors discovered in \cite{ShiKim2021}, such as the {\em snake sector} or $\mathbb{N}$ and $\mathbb{U}$-type sectors. In all examples checked by the authors, interpreting these sectors in terms of categorical data was possible using extensions of the pinching trick. For instance, to interpret the snake sector one performs two pinches, leaving the snake sector to become a disk through a composition of three domain walls (see Fig. \ref{snake}). As such, snake sectors correspond to the indecomposable bimodule subcategories of a triple relative tensor product. By pinching, the $\mathbb{N}$ and $\mathbb{U}$-type can be understood as excitations on composite domain walls. Perhaps there is a more elegant algebraic theory which can be used to simultaneously describe the space of sectors corresponding to all possible geometries,  
as well as all possible classes of interactions between these different types of sectors.

In this note, we created dictionaries among different approaches to studying domain walls and their excitations: category theory, QFT / generalized symmetries, and the entanglement bootstrap (see Tables~\ref{parton-classification} and \ref{parton-Gsymm} and others in the Appendices). Such dictionaries have yet to be written in three and higher (spatial) dimensions, where the categorical description of TQFT is also quite rich~\cite{Kong-Wen2014,Johnson-Freyd2022,Kong2024,kong2022invitation}. In higher dimensions, there are even bulk superselection sectors predicted by the entanglement bootstrap that are, as far as we know, yet to appear in the category theory literature (e.g., \emph{graph excitations}~\cite{immersion2023} in 3d). It is interesting to ask if the ideas presented in this work can be generalized to produce a dictionary encompassing such sectors.

Finally, as we have seen in our discussion in Sec. \ref{GenSym}, partons naturally arise in SymTFTs and can therefore be related to non-topological quantum field theories. It would be interesting to explore any constraints that partons and more general entanglement bootstrap sectors impose on renormalization group flows. This avenue of research may be particularly attractive given the fact that (appropriately generalized) $c$-theorems in various dimensions can be constructed using entanglement (e.g., see \cite{Casini:2022rlv} for a recent review).

\begin{figure}
\includegraphics[width=7.2cm]{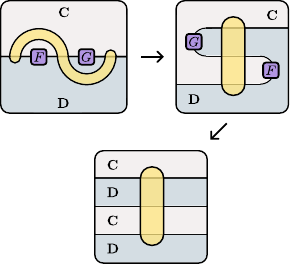}
\caption{Snake sectors of Ref.~\cite{ShiKim2021} can be understood in terms of domain wall fusion by applying the pinching trick twice.}
\label{snake}
\end{figure}

\section*{Acknowledgments}
MB thanks Mahesh Balasubramanian, Clement Delcamp, and Rajath Radhakrishnan for various discussions. BS thanks Ansi Bai and Yizhou Ma for independently proposing related questions after a conversation about entanglement bootstrap. MB's work was partly supported by the STFC under the grant, “Amplitudes, Strings and Duality.” RG acknowledges support by the Mani L. Bhaumik Institute for Theoretical Physics. BS was supported by the Simons Collaboration on Ultra-Quantum Matter, a grant from the Simons Foundation (652264, JM), the faculty startup grant of J. Y. Lee, the IQUIST fellowship at UIUC, and NSF award number PHY-2337931 at UC Davis. This research was supported in part by Perimeter Institute for Theoretical Physics and the International Centre for Mathematical Sciences. No new data were generated or analyzed in this study.

\bigskip\bigskip\bigskip

\appendix

\section{Partons in Witt non-trivial orders}\label{WittNT}
In this appendix, we touch upon the algebraic theory of partons for domain walls between Witt non-trivial topological orders (i.e., 2D topological orders not realizable by Levin-Wen models) whose bulk excitations are described by general uMTCs \footnote{More mathematically, we use the notion of Witt equivalence introduced in \cite{davydov2013witt,davydov2013structure}. In particular, we say a topological order, $\CP$, is Witt non-trivial if its corresponding Witt class is non-trivial (i.e., $\CP$ is not a Drinfeld center).}. The algebraic theory of partons in Witt-trivial topological orders (i.e., those realizable by Levin-Wen models) in the main text essentially describes partons as simple domain walls in the decomposition of a gapped domain wall and its dual. Such a notion is agnostic as to whether the theory is realizable by a Levin-Wen model and can therefore be defined for any topological field theory.

Notably, Witt non-trivial topological orders (including all chiral topological orders) are not expected to exactly satisfy the axioms of the entanglement bootstrap in a system of finite size ~\cite{Li2024}. However, we anticipate that a certain robust version of the entanglement bootstrap axioms can be used to describe Witt non-trivial states, from which partons may be derived rigorously.

An elegant and powerful framework for the analysis of domain walls and their fusion that covers the Witt non-trivial case was developed in \cite{huston2023composing}. This algebraic theory is slightly more involved, so we only sketch the relevant definitions. Ultimately, this construction relies on the realization of a Witt non-trivial topological order on the boundary of a $(3+1)$d Walker-Wang model. We assume that the bulk excitations of a $(2+1)$d topological order are described by a uMTC, $\CP$. In the approach of \cite{huston2023composing}, one focuses on a triple consisting of a unitary fusion category $\mathcal{X}$, a uMTC $\mathcal{A}$, and a braided tensor functor, $F:\mathcal{A}\to \CZ(\mathcal{X})$. These three pieces of data are subject to the condition $\CZ^{\mathcal{A}}(\mathcal{X})\cong \CP$, where $\CZ^{\mathcal{A}}$ is the Muger centralizer of the image of $\mathcal{A}$ under $F$ in $\CZ(\mathcal{X})$. Such a triple is called an $\mathcal{A}$-enriched fusion category $\mathcal{X}$ ($F$ is often omitted), and $\mathcal{A}$ plays the role of the Witt-trivializing theory (i.e., $\CZ(\CX)\cong \CA\boxtimes\CZ^{\CA}(\CX)$). For example, we can always take $\mathcal{A}=\CP^{rev}$ and $\mathcal{X}=\CP$ so that $\CZ(\CP)\cong \CP\boxtimes \CP^{rev}$. Crucially, whenever two topological orders admit a topological interface, they belong to the same Witt class, and we can present them using the same $\mathcal{A}$. 
  
 If $\CX$ and $\mathcal{Y}$ are $\mathcal{A}$-enriched fusion categories, then all topological domain walls between topological orders with bulk excitations $\CP\cong \CZ^{\CA}(\mathcal{X})$ and $\CQ\cong \CZ^{\CA}(\mathcal{Y})$ correspond to $\CA$-centered $\CX|\CY$-bimodule categories (see \cite{Brochier_2021} for the definition). The domain wall fusion is the relative Deligne product.

In this context, domain wall defects correspond to $\CA$-centered bimodule functors, $\Fun^{\CA}_{\CX|\CY}(\CM,\CN)$, defined as morphisms between $\CM$ and $\CN$ as objects in the Morita 4-category of braided fusion categories of \cite{Brochier_2021}. Such functors satisfy the Eilenberg-Watts theorem and an appropriate version of the  Frobenius reciprocity (\cite[Lemma 5.4]{Brochier_2021}) such that every $\CA$-centered bimodule functor, $F:\CN\to \CM$, corresponds to a bimodule functor, $\CX\to \CN^{*}\boxtimes_{\CX}\CM$.  Analogously to the Witt-trivial case, we identify partons with defects. We obtain a dictionary between partons and enriched fusion categories.
 \begin{table}[h]
\begin{tabular}{ |p{3.7cm}||p{3.7cm}|  }
 \hline
 Witt non-trivial topological order& Category theory\\
 \hline
 Bulk excitations   & $\CA$-enriched unitary fusion category $\CX$\\
 Gapped domain wall & $\CA$-enriched $\CX|\CY$- bimodule category  $\CM$\\
 Composition of domain walls  & Relative Deligne product \\
 $N$-type parton sectors & Indecomposable $\CA$-enriched $\CY|\CY$ bimodule subcategories of $\CM^{*}\boxtimes_{\CX}\CN$\\
 $U$-type parton sectors & Indecomposable $\CA$-enriched $\CX|\CX$ bimodule subcategories of $\CN\boxtimes_{\CY}\CM^*$\\
  $N$-type parton sector associated with a defect, F& $\, \mathfrak{n}_F\subseteq \mathcal{M}^* \boxtimes_{\mathcal{X}} \mathcal{N}$  \\
$U$-type parton sector associated with a defect, F  &  
 $\,\mathfrak{u}_F\subseteq \mathcal{N}\boxtimes_{\mathcal{Y}}\mathcal{M}^*$ \\
 \hline
\end{tabular}
\caption{Algebraic description of partons for gapped domain walls between Witt non-trivial orders}
\label{Anomalous}
\end{table}

There are several parallels between our problem and the one discussed in \cite{huston2023composing}: techniques developed there provide a computational tool for decomposition of a domain wall like $\CM\boxtimes_{\CY}\CM^*$. In that language, partons correspond to minimal projections in the convolution algebra of endomorphisms of the condensable algebra associated with $\CM$ (see the discussion in Section 5 of \cite{huston2023composing}). Intuitively, one can guess this answer by looking at Fig. \ref{parton-diagrams}~(b,c): we can interpet this diagram as a process of creation and annihilation of an excitation inside the yellow region. These excitations should be able to come out of the interface and come back in, and we can recognize them as endomorphisms of the condensable algebra.

%%%%%

\section{Some Illustrative Examples}\label{examples}

In this section, we study examples that illustrate the discussion in the main text. These examples are reinterpreted in an upcoming paper \cite{BalasubramanianXXX} from the perspective of gauging generalized symmetries.

\subsection{Partons and 1-form symmetry gauging}\label{app:exmp-non-chiral}
Let us first consider a class of examples in which the bulk TQFTs are $\CZ(\calC)$ and $\CZ(\CD)\cong \CZ(\CC)_A^{\rm loc}$. The latter Drinfeld center is the category of local modules over a condensable algebra, $A\in\CZ(\CC)$, and it arises via gauging a corresponding 1-form symmetry in $\CZ(\CC)$. As an object in $\CZ(\CC)$, we have
\begin{equation}\label{GenCond}
A=1+\sum_{\ell\ne1}\ell~,
\end{equation}
where each $\ell\in \CZ(\CC)$ has bosonic self-statistics, and $A$ also has a multiplication and unit map that allow us to consistently sum over it in the (2+1)$d$ spacetime. Performing this sum in half the spacetime produces $\CZ(\CD)$ and an interface with $\CZ(\CC)$. The theory of excitations on the domain wall between the two phases is isomorphic to the category of $A$ modules, $\CZ(\CC)_A$.

This class of examples is sufficiently general to illustrate two important properties of partons and their quantum dimensions:
\begin{enumerate}
\item Having partonic quantum dimension 1 is necessary but not sufficient for the corresponding surface to be invertible (or, in the entanglement bootstrap terminology, to be \lq\lq transparent").
\item Parton quantum dimensions (and partons themselves) are defined by pairs of topological phases separated by a domain wall rather than being intrinsic to a given topological phase.
\end{enumerate}

To better understand these points, let us first comment on the case in which we gauge a set of lines that form a ${\rm Rep}(G)\subset\CZ(\CC)$ fusion subcategory, where $G$ is a finite group. Here we have
\begin{equation}
A=\sum_{\ell\in {\rm Irrep}(G)}d_{\ell}\,\ell~,
\end{equation}
where $d_{\ell}$ is the quantum dimension of $\ell$ (which is equal to the dimension of the corresponding irrep of $G$ that is associated with $\ell$). In this case, the domain wall theory is the $G$-crossed extension of $\CZ(\CD)$ \cite{Barkeshli:2014cna}
\begin{equation}\label{Gcrossed}
\CZ(\CC)_A\cong\bigoplus_{g\in G}\CC_g~,
\end{equation}
where the identity component is $\CC_1\cong \CZ(\CC)_A^{\rm loc}\cong \CZ(\CD)$. The remaining components, $\CC_g$ with $g\ne1$, are non-local and consist of non-genuine lines attached to simple invertible surfaces implementing the action of the corresponding $g$ on the $\CZ(\CD)$ bulk. Therefore, group elements are in one-to-one correspondence with the $N$-type partons. Moreover, by \eqref{dn-and-du-def}, we have
\begin{equation}
\label{dn-RepG}
d^2_{\mathfrak{n}}=d^2_g=\frac{\sum_{F\in \mathcal{L}_O^{[g,{\bf 1}]}}d_{F}^2}{\sum_{F\in \mathcal{L}_O^{[{\bf1},{\bf1}]}} d_F^2}=1~,
\end{equation}
where the final equality follows from the fact that the surfaces satisfy the fusion rules of a finite group and that therefore the fusion of any element in $\mathcal{L}_O^{[g,1]}$ with any element in $\mathcal{L}_O^{[h,1]}$ must produce only elements in $\mathcal{L}_O^{[gh,1]}$\cite{Barkeshli:2014cna}. In other words, for ${\rm Rep}(G)$ condensation, all $N$-type partons are Abelian. Note that the partons in $\mathcal{L}_{\CO}^{[1,1]}$ can be pulled into the $\CZ(\CD)$ bulk without an attached surface and therefore have a notion of braiding associated with them (this statement also follows from our more general discussion in Sec. \ref{grading}).

Next, let us consider the $U$-type partons. By the discussion in \cite{BalasubramanianXXX,Buican:2023bzl},
\begin{equation}\label{SingleU}
S_A\cong \CN\boxtimes_{\CD}\CM^*~,
\end{equation}
is a single simple surface obtained by \lq\lq higher-gauging" $A$ \cite{Roumpedakis2023} (i.e., gauging a corresponding 1-form symmetry on a co-dimension one slice of spacetime). Therefore, there is a single $U$-type parton, $\mathfrak{u}=1$. It is also Abelian since, by \eqref{dn-and-du-def}, we have
 \begin{equation}
 \label{du-RepG}
 d^2_{\mathfrak{u}}=\frac{\sum_{F\in \mathcal{L}_O^{[{\bf1},{\bf 1}]}}d_{F}^2}{\sum_{F\in \mathcal{L}_O^{[{\bf1},{\bf1}]}} d_F^2}=1~.
 \end{equation}
 In summary, we learn that, for ${\rm Rep(G)}$ condensation, there are $|G|$ $N$-type partons, a single $U$-type parton, and all $|G|+1$ partons are Abelian.

Note that, even though $d_{\mathfrak{u}}=1$, $\mathfrak{u}$ is associated with a non-invertible surface. Indeed, thinking about this surface as an object in the 2-category $2{\rm Rep}(G)$, we obtain the fusion rule \cite{greenough2010monoidal}
\begin{equation}\label{SAfusion}
S_A\times S_A=|G|S_A~,
\end{equation}
where $|G|\ge2$ is the order of $G$. Therefore, this example illustrates the first point made in the list at the beginning of this subsection. Moreover, treated as a computation in 2-category theory (ignoring various QFT ambiguities), \eqref{SAfusion} implies a quantum dimension $d(S_A)=|G|$. As promised in the main text, this quantum dimension differs from the partonic quantum dimension \eqref{du-RepG}.

In the more general condensation case described around \eqref{GenCond}, the partons are no longer necessarily Abelian, but \eqref{SingleU} still holds by the same logic. Applying \eqref{du-RepG} to this more general case, we see there is always a single Abelian $U$-type parton when the two phases are related by anyon condensation in the way we have described.

The existence of non-Abelian $N$-type partons in more general condensations can be seen quite easily. For example, consider $\CZ(\CC)\cong D(S_3)$ and $\CZ(\CD)\cong D(\mathbb{Z}_2)$. As discussed in \cite{BalasubramanianXXX} (see also \cite{cui2019generalized}), the domain wall theory in this case has six objects, $\left\{X,Y\right\}\oplus\left\{1,e,m,f\right\}$ with
\begin{equation}
d_X = d_Y=2~, \ \ \ d_1=d_e=d_m=d_f=1~.
\end{equation}
Following \cite{BalasubramanianXXX}, we can interpret $X$ and $Y$ as boundary conditions for the $S_{1+e}$ condensation surface, and $1,e,m,f$ as boundary conditions for the trivial condensation surface (i.e., these correspond to $D(\mathbb{Z}_2)$ lines). Therefore, we have two $N$-type partons (labeled as $1$ and $\kappa$) and a single $U$-type parton. The non-trivial $N$-type parton corresponds to $S_{1+e}$. Applying the formulas in \eqref{dn-and-du-def}, we see that
\begin{equation}
d^2_{\mathfrak{n}}=\left\{1,2\right\}~, \ \ d^2_{\mathfrak{u}}=1~.
\end{equation}
Note that, as in the discussion around \eqref{SingleU}, the non-Abelian parton quantum dimension, $d_{\kappa}=\sqrt{2}$, differs from the naive 2-categorical quantum dimension, $d(S_{1+e})=2$. Moreover, if we consider the case $\CZ(\CC)\cong D(\mathbb{Z}_2)$ and $\CZ(\CD)\cong{\rm Vec}$ (i.e., the trivial TQFT), the $S_{1+e}$ surface corresponds to the single $U$-type parton. In this case it has quantum dimension 1. This example emphasizes the second point in the list at the beginning of this subsection: parton quantum dimensions (and partons themselves) are defined by pairs of topological phases separated by a domain wall rather than being intrinsic to a given topological phase. This fact is reinterpreted in Appendix~\ref{app:fused-wall-EB} from the perspective of the entanglement bootstrap.

%%%%%%

\subsubsection{Simple Chern-Simons Examples}\label{ChiralEx}
Our examples thus far all involve Drinfeld centers of spherical fusion categories, which are all non-chiral by construction (they admit gapped boundaries). However, our results apply more generally to chiral theories. Let us illustrate this point in a pedestrian manner for the simple case of the chiral Chern-Simons theory $\CC=SU(2)_{4k}$ (here $k\in\mathbb{Z}_+$). Thinking in terms of the examples discussed above, we can consider
\begin{equation}
\CZ(\CC)\cong SU(2)_{4k}\boxtimes\overline{SU(2)_{4k}}~.
\end{equation}
This bulk theory is non-chiral, and each factor is a Chern-Simons theory with anyons of $SU(2)$ spin $0$, $1/2$, $1$, $\cdots$, $2k$ \footnote{The topological spins for the two factors are complex conjugates of each other. In particular, we have
\begin{equation}
\theta_{\ell}=\exp\left(\pi i\ell(\ell+1)\over 2k+1\right)~, \ \ \ \theta_{\tilde\ell}=\exp\left(-{\pi i\tilde\ell(\tilde\ell+1)\over 2k+1}\right)~.
\end{equation}
Note that the $SU(2)$ spin $2k$ anyon is a boson.}. The lines in the theory can be labeled by their $SU(2)$ spin
\begin{equation}
(\ell,\tilde\ell)~,\ \ \ \ell,\tilde\ell\in\left\{0,1/2,1,\cdots,2k\right\}~.
\end{equation}
They have fusion rules
\begin{eqnarray}
(\ell_1,\tilde\ell_1)&\times&(\ell_2,\tilde\ell_2)=\cr&&\sum_{\ell=|\ell_1-\ell_2|, \tilde\ell=|\tilde\ell_1-\tilde\ell_2|}^{{\rm min}(\ell_1+\ell_2,4k-\ell_1-\ell_2),{\rm min}(\tilde\ell_1+\tilde\ell_2,4k-\tilde\ell_1-\tilde\ell_2)}(\ell,\tilde\ell)~.\ \ \ \ \ \
\end{eqnarray}
For these choices of levels, the lines $(2k,0)$, $(0,2k)$, and $(2k,2k)$ are bosons. Moreover, these lines form a $\mathbb{Z}_2\times\mathbb{Z}_2$ 1-form symmetry that acts in the following way via fusion
\begin{eqnarray}
(2k,0)\times(\ell,\tilde\ell)&=&(2k-\ell,\tilde\ell)~,\cr (0,2k)\times(\ell,\tilde\ell)&=&(\ell,2k-\tilde\ell)~,\cr (2k,2k)\times(\ell,\tilde\ell)&=&(2k-\ell,2k-\tilde\ell)~.
\end{eqnarray}

We may now consider gauging the $\mathbb{Z}_2\times\mathbb{Z}_2$ one-form symmetry by summing over the corresponding condensable algebra
\begin{equation}
A=(0,0)+(2k,0)+(0,2k)+(2k,2k)~,
\end{equation}
to produce the theory
\begin{equation}
\CZ(\CD)\cong SO(3)_{2k}\boxtimes \overline{SO(3)_{2k}}~.
\end{equation}
In this case, the domain wall theory is the $G=\mathbb{Z}_2\times\mathbb{Z}_2$-crossed theory
\begin{equation}
\CZ(\CC)_A=\CC_{1}\oplus\CC_{g_1}\oplus\CC_{g_2}\oplus\CC_{g_1g_2}~,
\end{equation}
where the simple objects in each factor are
\begin{eqnarray}
\CC_{1}&=&\left\{(\ell,\tilde\ell)|\ \ell,\tilde\ell\in\mathbb{Z}~,\ 0\le\ell,\tilde\ell<k\right\}\cr&\cup&\left\{(k,\tilde\ell,\pm)|\ 0\le\tilde\ell<k\right\}\cup\left\{(\ell,k,\pm)|\ 0\le\ell<k\right\}\cr&\cup&\left\{(k,k,\alpha,\beta)\right\} ~,\ \ \ \cr\CC_{g_1}&=&\left\{(\ell,\tilde\ell)|\ \ell={2n+1\over2}~, \  \tilde\ell\in\mathbb{Z}~,\ 0\le n,\tilde\ell<k~,\right\}~,\cr \CC_{g_2}&=&\left\{(\ell,\tilde\ell)|\ \tilde\ell={2\tilde n+1\over2}~, \  \ell\in\mathbb{Z}~,\ 0\le\tilde n,\ell<k~,\right\}~,\cr\CC_{g_1g_2}&=&\left\{(\ell,\tilde\ell)|\ \ell={2n+1\over2}~,\tilde\ell={2\tilde n+1\over2}~,\ 0\le n,\tilde n<k\right\},\ \ \ \ \ \ \ \nonumber
\end{eqnarray}
where $\alpha,\beta\in\left\{\pm\right\}$. Applying \eqref{dn-RepG} and \eqref{du-RepG} to the case at hand, we conclude that there are four Abelian $N$-type partons
\begin{equation}
d^2_{1}=d^2_{g_1}=d^2_{g_2}=d^2_{g_1g_2}=1~,
\end{equation}
and a single Abelian $U$-type parton.

To specify the theory on $\Sigma$ in the approach of Sec.~\ref{GenSym}, we need to give a Lagrangian algebra in $\CZ(\CC)$ and a Lagrangian algebra in $\CZ(\CD)$. A particularly simple choice is given by the canonical Lagrangian algebras
\begin{eqnarray}\label{LagAlg}
L_{\CZ(\CC)}&=&(0,0)+(1/2,1/2)+(1,1)+\cdots+(2k,2k)~,\cr L_{\CZ(\CD)}&=&(0,0)+(1,1)+\cdots+(k-1,k-1)+(k,k,+,+)\cr&\ &+(k,k,+,-)+(k,k,-,+)+(k,k,-,-)~.\ \ \ 
\end{eqnarray}
We can then \lq\lq unfold" the theory. The bulk theory is now an $SU(2)_{4k}$ Chern-Simons theory separated from an $SO(3)_{2k}$ Chern-Simons theory via a domain wall. There is now a surface operator, $S$, replacing the gapped boundary on $\Sigma$. For the case in \eqref{LagAlg}, the surface operator is trivial, $S=\mathds{1}$.

We can then work out the domain wall theory after unfolding since the domain wall theory before unfolding factorizes
\begin{equation}
\CZ(\CC)_A=(\CC_{1_+}\oplus\CC_{g_+})\boxtimes(\CC_{1_-}\oplus\CC_{g_-})~,
\end{equation}
where we factorize the theory into a factor arising from the domain wall between $SU(2)_{4k}$ and $SO(3)_{2k}$ and a second factor arising from the domain wall between $\overline{SU(2)_{4k}}$ and $\overline{SO(3)_{2k}}$. Then, we can perform the pinching trick before or after unfolding. This allows us to find the partons in the non-chiral and chiral theories. We have already determined the non-chiral partons before unfolding. To determine the parton content in the chiral case, we unfold and then pinch, and we expect to have two Abelian $N$-type partons (in addition to a single Abelian $U$-type parton corresponding to $S_{0+2k}$)
\begin{equation}
d_1^2=d_g^2=1~,
\end{equation}
where $g$ is associated with the diagonal $\mathbb{Z}_2\subset\mathbb{Z}_2\times\mathbb{Z}_2$ of the unfolded theory. Closely related results (not phrased in the language of partons) were arrived at for the case $k=1$ by the authors of \cite{huston2023composing} using the techniques sketched out in Appendix \ref{WittNT}.

\section{Wave functions, entanglement bootstrap, consistency rules}\label{EB}
 
In this appendix, we provide a crash course on the entanglement bootstrap view of gapped domain walls and also present some new results.  We focus on the axioms, parton quantum dimensions, and composite domain walls. We also summarize and derive a set of predicted rules that form a dictionary between the algebraic theory of partons and many-body wave functions. While the discussion remains mathematically accurate, some detailed proofs are omitted for brevity. All the figures in this section are in space, and the partitions are for a many-body wave function.

The starting point of the entanglement bootstrap is a vacuum (many-body reference state) of the domain wall that satisfies the axioms {\bf A0} and {\bf A1} in the bulk and the vicinity of the wall~\cite{ShiKim2021}. These axioms give rise to a \emph{simplex} of locally indistinguishable density matrices on $N$-shaped and $U$-shaped regions. The extreme points of such a simplex are the parton sectors. This construction underpins the intuition of pinching in Sec.~\ref{sec:pinching}.

\subsection{Parton quantum dimension in entanglement bootstrap}\label{app:dn-M-N}

We propose the following definition of parton quantum dimensions, $d_n$, for defects between different domain wall species {\bf M} and {\bf N}. Given a parton state, $\rho^n$, on the $N$-shape region $BCD$, we let
 \begin{equation}\label{eq:dn-EB-new}
    d_n :=  \exp{\left(\frac{\Delta(B,C,D)_{\rho^n}}{4}\right)},\,\, \vcenter{\hbox{\includegraphics[width=0.42\linewidth]{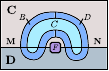}}}
\end{equation}
 where
 \begin{equation}
     \Delta(B,C,D) := S_{BC} + S_{CD} - S_B - S_D~,
 \end{equation}
 and $S(\rho_X)=-\Tr(\rho_X \log \rho_X)$ is the von Neumann entropy. Formula~\eqref{eq:dn-EB-new} is an analog of the anyon quantum dimension applicable to the figure-eight annulus~\cite{figure8}. Note that $d_n \ge 1$ follows immediately from this definition by the strong subadditivity of the von Neumann entropy.
 The motivation behind this definition is twofold:
\begin{enumerate}[leftmargin=13pt]
    \item This definition only needs a single quantum state. It does not require knowing the entropy difference between the parton state, $\rho^n$, and the vacuum, $\rho^1$, as in the original reference~\cite{ShiKim2021}. 
    \item This definition works when ${\bf M \ne N}$, and it broadens the scope of the original definition. For ${\bf M \ne N}$, there may not be a vacuum state (or even an Abelian state). Definition \eqref{eq:dn-EB-new} is equivalent to the original definition in the contexts where the original definition was proposed (i.e., when ${\bf M = N}$).
\end{enumerate}
In the main text, we discused the parton quantum dimension from the categorical perspective for ${\bf M = N}$ (see Sec.~\ref{sec:d_n} and especially Theorem~\ref{thm:d_n-M-M}). We leave a generalization of this analysis applicable to the more general definition~\eqref{eq:dn-EB-new} for the future.

\subsection{Fused domain walls}\label{app:fused-wall-EB}

From the viewpoint of the main text, the parton type and the domain wall type are closely related. The key formula
 \begin{equation}\label{eq:wall-fusion-rule}
  \calW^{{\bf D}\to {\bf C}} \times \left(\calW^{{\bf D}\to {\bf C}}\right)^\dagger = \bigoplus_{n\in \calC_N} \calW_n^{{\bf D\to D}}~,
\end{equation}
says that the possible indecomposable bimodule summands are in one-to-one correspondence with the parton type. 
We explain the intuition behind this formula from the viewpoint of quantum states. 

Suppose we are given a wave function of a pinched domain wall as in described in Fig.~\ref{pinching-trick}. We examine the pinched (or fused) domain wall wave function by studying its local entropy conditions, as shown in Fig.~\ref{fig:composite-wall-EB} (note that we may have ${\bf M} \ne {\bf N}$ as in the main text). If a single parton state is detected by pinching the domain wall sandwich, then one can verify certain entropy conditions in the partition of Fig.~\ref{fig:composite-wall-EB}(a), (namely $\delta_0=0$ and $\delta_1=0$), that guarantee the fused domain wall is well-defined in the sense of the entanglement bootstrap. A well-defined domain wall in the entanglement bootstrap sense should be understood as an indecomposable bimodule in category theory \footnote{It is also possible to have a many-body state corresponding to a superposition of multiple parton species. Such a state will violate axiom {\bf A0} of the entanglement bootstrap on the wall, (i.e., $\delta_0 > 0$ for Fig.~\ref{fig:composite-wall-EB}). It will not correspond to an indecomposable bimodule.}; see Table~\ref{Table:domain-wall-EB}.

\begin{figure}[h]
    \centering
    \includegraphics[width=0.99\linewidth]{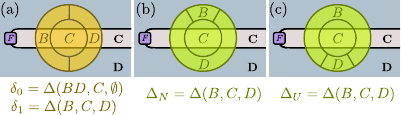}
    \caption{In accordance with the pinching trick Fig.~\ref{pinching-trick}, a pair of domain walls bent towards each other can be viewed as a fused wall at a large length scale. (a) Entropy conditions on the yellow partition $\delta_0$ and $\delta_1$ determine if the fused domain wall is well-defined. (b) and (c) The entropy conditions on the pair of green partitions $\Delta_N$ and $\Delta_U$ determine if the wall is transparent~\cite{ShiKim2021}.}
    \label{fig:composite-wall-EB}
\end{figure}

The next two quantities, $\Delta_{N}$ and $\Delta_U$, in Figs.~\ref{fig:composite-wall-EB}(b) and (c) respectively, are known to compute the total quantum dimensions of the parton species corresponding to the fused wall via $\exp({\Delta_{N}/2})$ and $\exp({\Delta_{U}/2})$. When both $\Delta_{N}$ and $\Delta_{U}$ vanish, the fused domain wall is \emph{transparent} in the sense of the entanglement bootstrap. (This notion corresponds to the statement that the domain wall is invertible in category theory language; see Table~\ref{Table:domain-wall-EB}.) Physically, the domain wall is transparent because, in this case, one cannot determine the location of the wall by entropy conditions, and the isomorphism theorem~\cite{shi2020fusion} guarantees the smoothness of passing information across the wall.  Importantly, a fused wall being transparent or not is a property understood at a sufficiently coarse-grained length scale. Analogously, a glass may be transparent to visible light, but it may not be transparent to UV radiation.

\begin{table}[h]
\begin{tabular}{ |p{4.0cm}||p{3.7cm}|}
\hline
Entanglement bootstrap & Category theory / generalized symmetries\\
 \hline
 Well-defined domain wall
 \begin{equation}
     \delta_0= \delta_1=0, \text{ in Fig.~\ref{fig:composite-wall-EB}} \nonumber
 \end{equation} & Indecomposable  bimodule / simple topological surface\\
 \hline  
 Transparent domain wall
 \begin{eqnarray}
     \delta_0 &=& \delta_1=0, \text{ in Fig.~\ref{fig:composite-wall-EB}} \nonumber\\
     \Delta_N &=& \Delta_U=0, \text{ in Fig.~\ref{fig:composite-wall-EB}} \nonumber
 \end{eqnarray}   & Invertible bimodule / invertible topological surface \,\,\,\, \\
 \hline  
\end{tabular}
\caption{ Mapping terminology regarding domain walls between the entanglement bootstrap and category theory / generalized symmetries.}
\label{Table:domain-wall-EB}
\end{table}

As we discussed in the main text, the relation between parton quantum dimensions and the properties of fused domain walls is subtle (see also Appendix~\ref{app:exmp-non-chiral} for concrete examples and discussions illustrating this point). In particular, we cannot determine the parton quantum dimensions from the fused walls alone. This subtlety is also reflected in the way we compute the parton quantum dimension via the wave function. None of the configurations in Fig.~\ref{fig:composite-wall-EB} are capable of computing $d_n$, and one needs to cut inside the fused wall as in Fig.~\ref{fig:composite-wall-dn}.

\begin{figure}[h]
    \centering
    \includegraphics[width=0.92\linewidth]{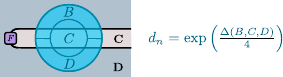}
    \caption{The partition that computes the quantum dimension of parton $n$. The regions $B,C,D$ are deformed versions of those in \eqref{eq:dn-EB-new}. The distinction between this partition and the one in Fig.~\ref{fig:composite-wall-EB} is that the here $B$ and $D$ know the details inside the narrow (pinched)  region between the original domain walls. Such fine detail is generally unavilable after coarse-graining to large length scales, and thus is not a property of the fused wall.}  
    \label{fig:composite-wall-dn}
\end{figure}

\subsection{Consistency rules}

In this final appendix, we translate a few predictions of the algebraic theory of partons in the main text to the world of many-body wave functions. Certain formulas previously derived within the entanglement bootstrap are justified by our algebraic theory. Note that the formulas in this section can be treated as predictions for many-body wave functions, and can, in principle, be tested. 

According to Fig.~\ref{fig:sandwich} and the discussion around it, the most elementary types of gapped domain walls are those obtained by anyon condensation. Other domain walls can be thought of as composites (up to the action of invertible surfaces). We shall focus on condensation domain walls throughout this appendix and assume that any defect we study does not change the domain wall type (i.e., ${\bf M}={\bf N}$). Note that, by a condensation domain wall, we mean a domain wall between ${\bf C}= \calZ(\calC)$ and ${\bf D}= \calZ(\calC)_A^{\text{loc}}$ obtained from gauging $A$ in half of spacetime. For chiral theories, condensation domain walls can be defined similarly (e.g., see  \cite{fuchs2013bicategories,Huston2022,davydov2013structure}).

Given a condensation domain wall, there is a unique $U$-type parton, and thus $\calD_{U}:= \sqrt{\sum_u d_u^2}=1$. The total quantum dimension of $N$-type partons, $\calD_{N}:= \sqrt{\sum_n d_n^2}$, obeys
\begin{equation}\label{eq:rule3}
     \sqrt{{\calD_{\bf C}}/{\calD_{\bf D}}} = {\calD_N}~,
\end{equation}
where $\CD_{\bf C}:=\sqrt{\sum_{a\in\calC_{\bf C}}d^2_a}$, and $\CD_{\bf D}:=\sqrt{\sum_{x\in \calC_{\bf D}}d^2_x}$. This formula follows from \eqref{dn-and-du-def}, which we reproduce below for completeness
\begin{equation} \label{eq:rule-grading}
d_n^2 = {\sum_{\alpha \in \mathcal{L}^{[n,1]}_O} d_\alpha^2\over\sum_{\alpha \in \mathcal{L}^{[1,1]}_O} d_\alpha^2}~.
\end{equation}
Indeed, combining this equation with the fact that $\sum_{F\in\mathcal{L}_O^{[1,1]}}d^2_F=\CD^2_{\bf D}$, and the fact that $\sum_{n,F\in\mathcal{L}_O^{[n,1]}}=\CD_{\bf C}\CD_{\bf D}$, we arrive at \eqref{eq:rule3}. 
In fact, it is also true that~\cite{Frohlich:2003hm}
\begin{equation}\label{FPdimRel}
     d_A: =\sum_{a\in A}N_ad_a = {\calD_{\bf C}}/{\calD_{\bf D}}~.
\end{equation}
Here, $d_a$ is the quantum dimension of anyon $a\in{\bf C}$, $N_a$ is the dimension of the junction space between anyon $a$ and the domain wall, and 
 $d_A$ is the quantum dimension of the object underlying the condensation algebra, $A= \sum_a N_a a$.
Together with Eq.~\eqref{eq:rule3}, we get
\begin{equation}\label{eq:rule1}
    \sum_{n\in \calC_N} d_n^2 = \sum_{a\in A} N_a d_a~.
\end{equation}
This formula was also derived in the entanglement bootstrap in \cite{ShiKim2021}.

Finally, another basic computation in the entanglement bootstrap using pairing manifolds~\cite{immersion2023} shows that the number of $N$-type partons can be counted via the $N_a$ described above
\begin{equation}\label{eq:rule2}
|\calC_N|  = \sum_{a \in A} (N_a)^2~. 
\end{equation}
We can separately argue for this formula using the identification between parton sectors and possible fusion channels of domain walls, as demonstrated in Fig. \ref{Ashrink} and Fig. \ref{provingB3}.

\begin{figure}[h]
    \centering
    \includegraphics[width=0.95\linewidth]{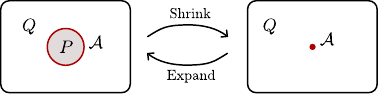}
    \caption{Consider a spatial slice with two topological phases, $Q$ and $P\cong Q_A^{\rm loc}$, separated by a domain wall, $\CA$ (we abuse notation and use $\CA$ to also label the condensable algebra in $Q$ that we sum over to produce $P$). If we shrink this domain wall, then we obtain a (non-simple) anyon, $\CA \in Q$. We can also go in the opposite direction and expand the point excitation, $\CA$, into a domain wall between $Q$ and $P$. See \cite{Kaidi:2021gbs} for the case in which $P$ is trivial (i.e., the domain wall is a gapped boundary).}
    \label{Ashrink}
\end{figure}

\begin{figure}[h]
\centering
\includegraphics[width=0.98\linewidth]{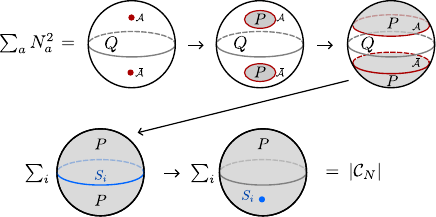}
\caption{The proof of \eqref{eq:rule2}. We consider the 3-manifold $X\cong S^2\times S^1$ (the $S^1$ is not depicted). Inserting the anyons $\CA$ and $\bar\CA=\CA$ from Fig.~\ref{Ashrink} computes the RHS of \eqref{eq:rule2} by fusing $\CA\times\bar\CA$. Alternatively, we can  expand $\CA$ and $\bar\CA$ as in Fig. \ref{Ashrink} and then fuse the corresponding domain walls. This maneuver produces the sum over indecomposable $S_i$ surfaces in $Q$. These surfaces can, in turn, be shrunk to points where we produce anyons, $A_i:=S_i(1)$, that satisfy $\dim({\rm Hom}(A_i,1))=1$ (this latter statement can be seen from folding and considering the corresponding Lagrangian algebra). Since only the unit contributes from each $S_i(1)$, and we are left with $|\CC_N|$ copies of the $S^2\times S^1$ partition function. We have established \eqref{eq:rule2}.}
\label{provingB3}
\end{figure}

\bibliographystyle{ucsd}
\bibliography{ref}  
\end{document}